\documentclass[onecolumn, a4size, 11pt]{IEEEtran}
\usepackage{amsmath}
\usepackage{amssymb}
\usepackage{amsfonts}
\usepackage{graphicx}
\usepackage{epsfig}
\usepackage{epstopdf}
\usepackage{subfigure}
\usepackage{psfrag}
\usepackage{cite}
\usepackage{url}
\usepackage{color}
\usepackage{bigstrut} 
\usepackage{multirow} 
\usepackage{array}
\usepackage[dvipsnames]{xcolor}

\newcolumntype{c}{>{\centering\arraybackslash}p{8em}}
\newcolumntype{N}{>{\centering\arraybackslash}p{5em}}
\newcolumntype{D}{>{\centering\arraybackslash}p{7em}}
\newcolumntype{S}{>{\centering\arraybackslash}p{10em}}
\newcommand{\real}{\text{Re}}
\title{Multiuser Wireless Power Transfer via Magnetic Resonant Coupling: Performance Analysis, Charging Control, and Power Region Characterization}
\author{Mohammad R. Vedady Moghadam, \textit{Member}, \textit{IEEE}, and Rui Zhang, \textit{Senior Member}, \textit{IEEE}
\thanks{This paper was presented in part at IEEE International Conference on Acoustics, Speech, and Signal Processing (ICASSP), Brisbane,
Australia, April 19-24, 2015 \cite{ReZa}.}
\thanks{M. R. Vedady Moghadam is with the Department of Electrical and Computer Engineering, National University of Singapore, Singapore 117583  (e-mail: vedady.m@u.nus.edu).}
\thanks{R. Zhang is with the Department of Electrical and Computer Engineering, National University of Singapore, Singapore 117583 (e-mail: elezhang@nus.edu.sg).  He is also with the Institute for Infocomm Research, A*STAR, Singapore 138632.}
}
\begin{document}
\maketitle \thispagestyle{empty} 
\begin{abstract}
Magnetic resonant coupling (MRC) is an efficient method for realizing the  near-field  wireless power transfer (WPT). 
Although the MRC enabled WPT (MRC-WPT) with a single pair of transmitter and receiver has been thoroughly  studied in the literature, there is limited  work on the general setup  with  multiple  transmitters and/or receivers. 
In this paper, we consider a \textit{point-to-multipoint} MRC-WPT system with  one transmitter delivering  wireless  power to a set of distributed receivers. 
We aim to introduce new applications of signal processing and optimization techniques to the performance characterization and optimization in  multiuser WPT via MRC.  
We first derive closed-form expressions for the power  drawn from the energy source at the transmitter and that delivered to  the load at each receiver. 
We identify a \textit{``near-far''} fairness issue in multiuser power transmission due to receivers' distance-dependent mutual inductance with the transmitter.
To tackle this issue,  we propose a centralized \textit{charging control}  algorithm to jointly optimize the receivers' load  resistance to minimize the total transmitter power drawn while meeting  the given power requirement of each individual load.  
For ease of practical implementation, we also devise a \textit{distributed} algorithm for the receivers to adjust their load resistance independently in an iterative manner. 
Last, we characterize the \textit{power region} that constitutes all the achievable power-tuples of the loads via controlling their adjustable resistance. In particular, we compare the power regions without versus with the \textit{time sharing} of users' power transmission, where it is shown that time sharing yields  a larger power region in general.  
Extensive simulation results are provided to validate our analysis and corroborate our study on the multiuser MRC-WPT system.  
\end{abstract}
\setlength{\baselineskip}{1.0\baselineskip}
\newtheorem{definition}{\underline{Definition}}[section]
\newtheorem{fact}{Fact}
\newtheorem{assumption}{Assumption}
\newtheorem{theorem}{\underline{Theorem}}[section]
\newtheorem{lemma}{\underline{Lemma}}[section]
\newtheorem{corollary}{\underline{Corollary}}[section]
\newtheorem{proposition}{\underline{Proposition}}[section]
\newtheorem{example}{\underline{Example}}[section]
\newtheorem{remark}{\underline{Remark}}[section]
\newtheorem{algorithm}{\underline{Algorithm}}[section]
\newcommand{\mv}[1]{\mbox{\boldmath{$ #1 $}}}
\section{Introduction} \label{sec:introduction}
Inductive coupling  \cite{Murakami,Kim} is a traditional method to realize the near-field wireless power transfer (WPT) for short-range applications in e.g., centimeters. Recently,  magnetic resonant coupling (MRC)  \cite{Kurs, Fei, Shin, Chen} has drawn significant interest for implementing the near-field WPT due to its high power transfer efficiency as well as long operation range, say, up to a couple of meters. Furthermore,  MRC effectively avoids  the power leakage to non-resonant externalities  and thus ensures safety to the neighboring environment. 

Two different methods are commonly adopted in practice to implement   MRC enabled WPT (MRC-WPT). 
In the first method  \cite{Kurs,Fei}, resonators,  each of which is a tunable RLC circuit,  are  placed in close proximity of  the electromagnetic (EM) coils of the energy transmitters and receivers to efficiently transfer power between them. 
Since resonators are designed to resonate at the system's operating frequency,  the total reactive power consumption in the system is effectively minimized at resonance   and hence high power transfer efficiency is achieved over longer distance as compared to conventional  inductive coupling. 
In the second method  \cite{Shin,Chen}, series and/or shunt compensators, each of which  is a capacitor of variable capacity, are embedded in the electric circuits of energy transmitters  and receivers with  their natural  frequencies  set same as the system's operating frequency to achieve resonance.   
Generally speaking, the second method  achieves higher  power transfer efficiency over the first method, since in the first method  resonators   incur additional power loss due to their parasitic resistance. 
However,  the electric circuits of energy transmitters and receivers need to be  accessible  in the second method to embed compensators in them.  

The MRC-WPT system with a single pair of transmitter and receiver has been extensively studied in the literature,  with the aims such as  maximizing  the   end-to-end power transfer efficiency or  maximizing the power   delivered to the receiver's load with a given input power \cite{Cannon,Jonah, YZhang1,YZhang2}. 
Moreover,  systems with  two transmitters and a single receiver or with a single transmitter and two receivers have been studied in \cite{Yoon,K-Lee, Ahn, Garnica, Johari}, while their results cannot be directly applied to the systems with more than two transmitters/receivers.  Recently, an MRC-WPT system with multiple transmitters and one single receiver  has been  investigated in \cite{JD} to wirelessly charge a cellphone located at $40$ centimeters away, independent of the phone's orientation. However,  the interactions between  the energy transmitters and receiver  were demonstrated only  through   simulations in \cite{JD}. 
There have been other recent works (see e.g. \cite{Waters,Lang})   on optimizing the performance of  MRC-WPT systems with multiple transmitters and one single receiver.

Different from the above works, in this paper we  consider a \textit{point-to-multipoint} MRC-WPT  system based on the series compensator method aforementioned, as shown in Fig. \ref{fig:ElecCirtuit}, where one transmitter that is connected to a stable energy source supplies wireless power to a set of distributed receivers. Each receiver is connected to an electric load via a switch, where the switch connects/disconnects the load to/from the receiver. 
\begin{figure} [t]
\centering
\includegraphics[width=12cm]{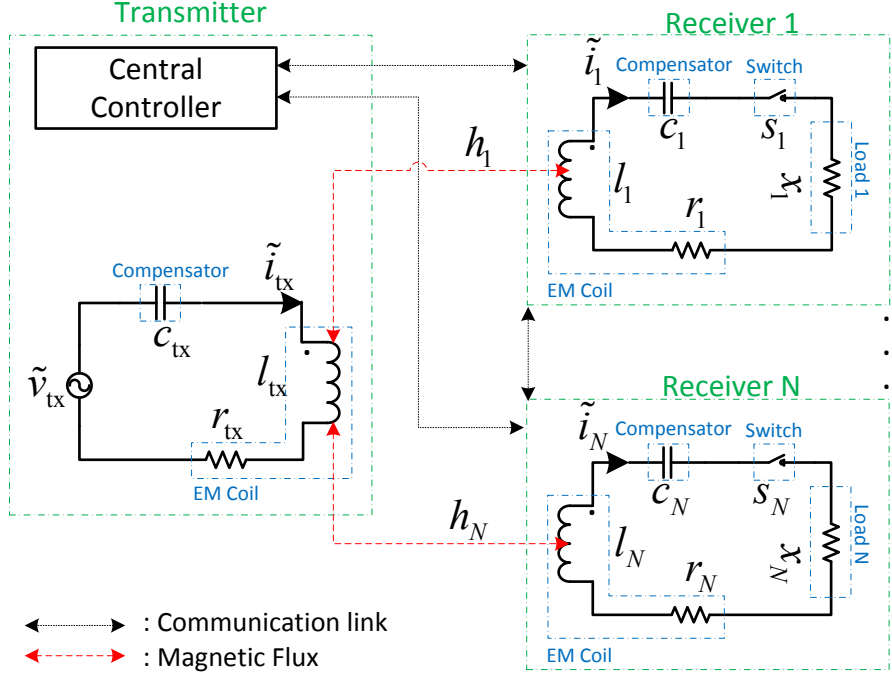} 
\caption{A point-to-multipoint magnetic resonant coupling enabled wireless power transfer system  with communication and control.} 
\label{fig:ElecCirtuit} 
\end{figure} 
We aim to apply signal processing and optimization techniques to the performance characterization and optimization in multiuser MRC-WPT systems. First, by extending the results in  \cite{Yoon,K-Lee, Ahn, Garnica, Johari,JD}, we  derive  \textit{closed-form} expressions for the   power drawn from the energy source at the transmitter and that  delivered to the load at each receiver, in terms of  mutual inductance among the transmitter and receivers as well as their circuit  parameters, for arbitrary number of receivers. 
Our obtained results reveal  a  \textit{near-far} fairness issue in  multiuser wireless power transmission, similar to its counterpart phenomenon in multiuser wireless communication. 
Specifically,  a receiver that is far  from the transmitter and thus has a small mutual inductance with the transmitter   receives  lower power as compared to a receiver that is closer to the transmitter, with other circuit parameters given identical. 
Next, we propose a method  to mitigate  the near-far issue  by jointly designing the load resistance of all receivers to control their received power by exploiting the mutual  coupling effect in the MRC-WPT system. This is analogous and yet in sharp contrast  to the method of adjusting antenna weights at the transmitter  to control the received power at different receivers in the existing far-field microwave or radio frequency (RF)  transmission enabled WPT \cite{ZHANG,Xu}. 

In particular, we consider the scenario where  a central controller is equipped at the transmitter to coordinate the multiuser power charging, by assuming that it has the full knowledge of all receivers, including their circuit parameters and power requirements. The central controller \textit{jointly} designs the adjustable load resistance of all receivers to minimize the total power consumed at the transmitter subject to the given minimum  power requirement of each load. 
For ease of practical implementation, we also consider the scenario without any central controller installed and devise a \textit{distributed} algorithm for multiuser charging control by  adjusting the loads' resistance at their individual receivers in an iterative manner. In our proposed distributed algorithm, each receiver sets its load resistance independently based on its local information and a one-bit feedback broadcasted  by each of the other receivers. The feedback of each receiver indicates whether the received power of its load exceeds the required minimum  power level or not. It is shown via simulations   that the distributed algorithm achieves performance fairly close to  the optimal solution by  the centralized algorithm with a finite number of iterations. 

Last,  we characterize the \textit{power region} for multiuser power transfer which constitutes all the achievable power-tuples for the receiver loads via controlling  their adjustable  resistance in given ranges. Specifically, we introduce the \textit{time-sharing} based multiuser power transfer, where the transmission is divided into orthogonal time slots and within each time slot only a selected subset of receivers are scheduled to  receive power, while the other receivers are disconnected from their loads. This is aimed to more flexibly  control the mutual coupling effect between the transmitter and receivers  in WPT.   
It is shown that time sharing can enlarge  the power region over the case without time sharing in general.  It is also shown that time sharing  can further mitigate the near-fare issue in multiuser WPT  by allocating more time  to  receivers that are more far-away from the transmitter.
Furthermore, we extend the  centralized multiuser charging control algorithm for the case without time sharing to jointly optimize the time allocation and  load resistance for all the receivers in the case with time sharing, to further reduce the transmitter power consumption under the same average power requirement of each load. 

The rest of this paper is organized as follows. Section II
introduces  the system model. 
Section III presents our  analytical  results. 
Section IV presents both the centralized and distributed multiuser power charging control algorithms.  
Section V  characterizes and compares the power regions without versus with time sharing. Finally, we conclude the paper in Section VI.
\section{System model} \label{sec:intro} 
As shown in Fig. \ref{fig:ElecCirtuit}, we consider an  MRC-WPT system with  a single transmitter and $N\ge1$ receivers, indexed by $n$, $ n \in {\cal N}=\{1,\ldots,N\}$.  
The transmitter and receivers are equipped with EM coils for realizing wireless power transfer, while an embedded communication system is assumed to enable information exchange among them.\footnote{As an example, the alliance for wireless power (A4WP) specification \cite{Nadakuduti} uses a  low energy  profile Bluetooth network at the band of $2.4$GHz for communication and  system control, which is aimed to schedule the  charging sequence of receivers and also control their individual charging power   according to the given  priorities.} 
The transmitter is connected to a stable energy source supplying sinusoidal voltage over time given by  $\tilde{v}_{\text{tx}}(t)=\real\{ {v}_{\text{tx}} e^{j w t} \}$,  with  ${v}_{\text{tx}}$ denoting a complex voltage which
is assumed to be constant, and $w>0$ denoting its   operating angular frequency.  
Each receiver $n$ is connected via a switch  to a given electric load (e.g.,  battery charger),  named   load $n$,  with adjustable resistance  $x_{n} > 0$. The switch is used to  connect/disconnect each  load to/from its corresponding receiver.
The state of switch at each receiver $n$ is given by $s_n \in \{0,~1\}$, where $s_n=1$ and $s_n=0$ denote the switch is closed and open, respectively.
It is also assumed  that the transmitter and each receiver $n$ are compensated  using series capacitors with capacities  $c_{\text{tx}} > 0$ and $c_{n} > 0$, respectively. 

Let $\tilde{i}_{\text{tx}}(t)=\real\{i_{\text{tx}}e^{j w t}\}$,  with complex-valued   ${i}_{\text{tx}} $, denote the steady state current  flowing through the transmitter. This current produces a time-varying magnetic flux in the transmitter's EM  coil, which passes through the EM  coils of nearby receivers  and induces time-varying currents in them. 
We  denote  $\tilde{i}_{n}(t)=\real\{i_{n}e^{j w t}\}$,  with complex-valued  ${i}_{n} $, as the steady state current at receiver $n$. It is worth pointing out that  the magnetic flux is the main medium of wireless power transfer considered in this paper, while the electric field is evanescent and thus is ignored \cite{Kurs}.   This is in contrast to the RF based far-field WPT \cite{ZHANG,Xu}, where the synchronized oscillations of magnetic and electric fields radiate  energy in the form of EM waves propagating through  the air.

We denote  $r_{\text{tx}}>0$ ($r_{n}>0$)  and $l_{\text{tx}}>0$ ($l_{n}>0$)  as the internal resistance and the self-inductance of the EM coil of the transmitter (receiver $n$), respectively. 
We also denote the mutual inductance between EM coils of the transmitter and each receiver $n$ by a real number $h_{n}$, with  $|h_{n}|  \le  \sqrt{l_nl_{\text{tx}}}$, where its actual value depends on the physical characteristics of the two EM coils, their locations, alignment (or misalignment) of  oriented axes with respect to each other, the environment magnetic permeability, etc.
For example, the mutual inductance of two coaxial circular loops that lie in the parallel planes with separating distance of $d$ meter is shown to be proportional to $d^{-3}$  in \cite{Cheng}.  
Moreover, since the receivers usually employ smaller EM coils than that of the transmitter due to practical size limitation and they are also physically separated, we ignore the mutual inductance between any pair of the receivers for simplicity. 

The equivalent electric circuit model of the considered MRC-WPT system is  also shown in Fig. \ref{fig:ElecCirtuit},
in which the natural angular frequencies of the transmitter and  each receiver $n$ can be expressed as $w_{\text{natural},\text{tx}}={1}/{\sqrt{l_{\text{tx}} c_{\text{tx}}}}$ and $w_{\text{natural},n}={1}/{\sqrt{l_{n} c_{n}}}$, respectively. 
We thus set the capacities of compensators' capacitors  as
\begin{align}
c_{\text{tx}}&=\dfrac{1}{l_{\text{tx}} \hspace{.6mm} w^{2}}, \label{eq:c0}\\
c_{n}&=\dfrac{1}{l_{n}\hspace{.6mm} w^{2}}, ~\forall n \in \cal N, \label{eq:cn}
\end{align} 
so  that the transmitter and all receivers have the same natural angular frequency as  the transmitter voltage source's  angular frequency $w$,  i.e.,   $w_{\text{natural},\text{tx}}=w_{\text{natural},1}=\ldots=w_{\text{natural},N}=w$. Accordingly, we name $w$ as the \textit{resonant angular frequency}.  

In this paper, we assume that the transmitter and all receivers are at fixed positions and the physical characteristics of their EM coils are \textit{a priori} known. As a result, $h_{n}$'s, $\forall n\in \cal N$, are modeled as  given constants, which are computed according to  Appendix E. In practical systems with mobile receivers,  $h_n$'s in general change over time and thus need to be measured periodically.  For example, one method that can be used in practice to estimate the mutual inductance between the transmitter and any receiver $n$, is given as follows. First, by disconnecting the loads at all other receivers $k\neq n$, under a known input voltage $v_{\text{tx}}$, the transmitter measures the power drawn from its voltage source, denoted by $p_{\text{tx}}$,  due to load $n$ only. From (\ref{eq:PT}), we can show
\begin{align}
h_n=\pm \dfrac{\sqrt{\left(\dfrac{|v_{\text{tx}}|^2}{2p_{\text{tx}}} -r_{\text{tx}} \right) \left(r_n+x_n\right)}}{w},
\end{align}
i.e., the transmitter can obtain the mutual inductance with receiver $n$ by assuming known $r_n$ and $x_n$ (which can be sent to the transmitter via one-time feedback from receiver $n$). Note that the sign of  $h_n$  can be determined by comparing the known direction of the current flowing in receiver $n$ (via a one-bit feedback from receiver $n$) with that assumed at the transmitter. If the directions are same, then the positive sign is selected for $h_n$; otherwise, the negative sign is set.

In this paper, we treat the load  resistance $x_{n}$'s, $\forall n \in \cal N$,  as design parameters,  which can be adjusted in real time  to control the performance of our considered  MRC-WPT system  based on the information shared among different nodes in the system, via the embedded communication system.
Note that an  electric load  with any  fixed resistance can be connected via a rectifier in parallel with a boost (or triboost) converter to each receiver  to realize an adjustable  resistance \cite{Pantic}. 
Specifically,   given the fixed input voltage, the on/off  time intervals of the  converter can be controlled in real time to change  the average current flowing into the load, which is equivalent to adjusting the load resistance. 
\section{Performance Analysis}  \label{sec:Performance}
In this section, we  present  new analytical results on the performance of the MRC-WPT system with arbitrary number of receivers. 
A numerical example is also  provided  to validate our analysis and draw useful insights. Here,  we assume that all receiver switches are closed, i.e., $s_n=1$, $\forall n \in \cal N$; as a result, the transmitter sends wireless power  to all loads  concurrently.
\subsection{Analytical Results}
By applying Kirchhoff's circuit laws to the electric circuit model given in Fig. \ref{fig:ElecCirtuit}, we have  
\begin{align}  
\hspace{-2mm}\left(\hspace{-.5mm}r_{\text{tx}} + j \left(\hspace{-.5mm}w l_{\text{tx}} \hspace{-.5mm}-\hspace{-.5mm}\dfrac{1}{w c_{\text{tx}}} \hspace{-.5mm}\right)  \hspace{-.5mm}\right) i_{\text{tx}}- j w \sum_{k \in \cal N} h_k i_k& = v_{\text{tx}}, \label{eq:TRANS} \\
\hspace{-2mm}\left(\hspace{-.5mm}r_n+x_n+j\left(\hspace{-.5mm}w l_n -\dfrac{1}{w c_n}\hspace{-.5mm}\right) \hspace{-.5mm}\right) i_n\hspace{-.5mm}-\hspace{-.5mm}j w h_n i_{\text{tx}}&=0, ~\hspace{-.5mm}\forall n \in \cal N.\hspace{-.5mm} \label{eq:RECIV}
\end{align}
From (\ref{eq:c0}) and (\ref{eq:cn}), we can set $w l_{\text{tx}} -{1}/{(w c_{\text{tx}})}=0$ and $w l_n -{1}/{(w c_n)}=0$ in (\ref{eq:TRANS}) and (\ref{eq:RECIV}), respectively. This is due to the fact that the transmitter and all receivers  are designed to resonate at the same angular frequency $w$.
By solving the set of linear equations given in (\ref{eq:TRANS}) and (\ref{eq:RECIV}), we can derive $i_{\text{tx}}$ and $i_n$'s as functions of  the input voltage $v_{\text{tx}}$ as  follows:
\begin{align} 
i_{\text{tx}}&=~\dfrac{1}{r_{\text{tx}}+ w^2 \sum_{k \in \cal N} h_{k}^2\left(r_{k}+x_{k}\right)^{-1} }v_{\text{tx}}, \label{eq:IT1m} \\
i_{n}&=\hspace{-0.2mm} j \hspace{-0.2mm}\dfrac{w h_n\left(r_{n}+x_{n}\right)^{-1} }{r_{\text{tx}}+ w^2 \sum_{k \in \cal N} h_{k}^2\left(r_{k}+x_{k}\right)^{-1} } v_{\text{tx}}, ~ \forall n \in \cal N. \label{eq:ILm}
\end{align}
The power drawn from the energy source at the transmitter, i.e.,  $p_{\text{tx}}$, and that delivered to each load $n$,  denoted by $p_{n}$, are then obtained  as
\begin{align} \label{eq:PT}
p_{\text{tx}}&= \dfrac{1}{2}\real \left\{v_{\text{tx}} i_{\text{tx}}^{*}\right\}
=\dfrac{|v_{\text{tx}}|^2}{2} \dfrac{1}{r_{\text{tx}}+ w^2 \sum_{k \in \cal N} h_{k}^2\left(r_{k}+x_{k}\right)^{-1}}, \\
p_{n}&=\dfrac{1}{2} x_{n} \left|i_{n}\right|^2 \label{eq:PL} =\dfrac{|v_{\text{tx}}|^2}{2} \dfrac{w^2 h_n^2 x_{n} \left(r_{n}+x_{n}\right)^{-2}}{\left(r_{\text{tx}}+ w^2 \sum_{k \in \cal N} h_{k}^2\left(r_{k}+x_{k}\right)^{-1} \right)^2},
\end{align}
where $i_{\text{tx}}^{*}$ denotes the conjugate of $i_{\text{tx}}$. From (\ref{eq:PL}), it follows that the power delivered to each load $n$ increases with the mutual inductance between EM coils of its receiver and the transmitter, i.e., $h_n$. 
This can cause a \textit{near-far} fairness issue since a receiver that is far  from the transmitter generally  has a small mutual inductance with the transmitter; as a result, its received power is lower than that at a receiver that is closer to the transmitter  (thus has a larger mutual inductance). 
Furthermore, we define  $p_{\text{sum}}=\sum_{k=1}^{N}p_{k}$  as the \textit{sum-power} delivered to all loads, where it can be verified from (\ref{eq:PT}) and (\ref{eq:PL}) that $p_{\text{sum}}< p_{\text{tx}}$.  
The  sum-power \textit{transfer efficiency}, denoted by $0\le \rho < 1$, is thus expressed as
\begin{align} \label{eq:eta}
\rho
=\dfrac{p_{\text{sum}}}{p_{\text{tx}}}= \dfrac{ w^2\sum_{k\in \cal N} h_k^2 \hspace{.3mm} x_{k} \hspace{.3mm}\left(r_{k}+x_{k}\right)^{-2}} {r_{\text{tx}}+ w^2 \sum_{k \in \cal N} h_{k}^2\left(r_{k}+x_{k}\right)^{-1}}.
\end{align}

\begin{remark} \label{Remark:1}
When  the receivers are all weakly coupled to the transmitter, e.g.,  they are sufficiently far away from the transmitter, we have  $h_n \rightarrow 0$, $\forall n \in \cal N$. In this regime,  from (\ref{eq:PT}),  it follows that the transmitter power is $p_{\text{tx}}\approx |v_{\text{tx}}|^2/(2 r_{\text{tx}})$, which is a function of  the resistance and  voltage of the transmitter only. 
On the other hand, from (\ref{eq:PL}), it follows that the power delivered to each load $n$ is $p_{n} \approx |v_{\text{tx}}|^2 w^2 h_n^2 x_n(r_n+x_n)^{-2}/ (2 r_{\text{tx}}^2)$, which is  irrespective  of the other receivers' mutual inductance and resistance.
The above results  can be explained as follows. With  $h_n \rightarrow 0$, $\forall n \in \cal N$, the power transfered to the receivers is small and thus can be neglected as compare to  the power loss due to the  transmitter's resistance. As a result,  we have   $p_{\text{tx}} \approx r_{\text{tx}}|i_{\text{tx}}|^2/2$, with   $i_{\text{tx}}=v_{\text{tx}}/r_{\text{tx}}$.   
It also can be  verified that  with $h_n \rightarrow 0$, $\forall n \in \cal N$,   the coupling effect among the receivers through the transmitter current $i_{\text{tx}}$ is negligible. 
Hence, the power delivered to the load at  each receiver is independent of other receivers (similar to the far-field RF based WPT \cite{ZHANG,Xu}).
\end{remark}
\begin{remark} \label{Remark:w}
It can be shown from (\ref{eq:PL}) that $p_n$, $\forall n \in \cal N$, first increases over $0<w<\dot{w}$, and then decreases over $w>\dot{w}$,  where 
\begin{align} \label{eq:w_Star}
\dot{w}=\sqrt{\dfrac{r_{\text{tx}}}{\sum_{k\in \cal N}h_k^2\left(r_k+x_k\right)^{-1}} } .
\end{align}
The above result can be explained as follows. From (\ref{eq:ILm}), it follows that  the magnitude of the current flowing in each receiver $n$, i.e., $|i_n|$, strictly increases  over $0<w<\dot{w}$, but  strictly decreases  over $w>\dot{w}$. This yields that $w=\dot{w}$ is the unique maximizer of $|i_n|$ over $w>0$. 
Obviously, $p_n$, which is defined in (\ref{eq:PL}) as $p_n=x_n |i_n|^2/2$,  behaves  same as $|i_n|$ over $w>0$. Although $w$ is assumed to be fixed in this paper,  it can also be optimally set to maximize the system power transfer efficiency, if this is implementable in practice. Furthermore, from (\ref{eq:w_Star}), it follows that $\dot{w}$ depends on the distances between the transmitter and receivers, since $h_n$'s in general  decrease with larger distances.
\end{remark}

Next, we study the effect of  changing the load resistance of one particular receiver $n$, i.e., $x_n$, on  the transmitter power  $p_{\text{tx}}$, its received power $p_n$, the power delivered to each of the other loads $m \in \cal N$, $m\neq n$, i.e., $p_m$,  the sum-power delivered to all loads $p_{\text{sum}}$, and the sum-power transfer efficiency $\rho$,  assuming  that  all the other loads' resistance  is  fixed. 
\begin{proposition} \label{Prop:3}
$p_{\textnormal{tx}}$ strictly increases over $x_n>0$.
\end{proposition}
\begin{proof}
Please see Appendix A.
\end{proof}

This result  can be explained as follows.   From  (\ref{eq:IT1m}), it is  observed that the transmitter current $|i_{\text{tx}}|$  strictly increases over  $x_n>0$. Since the energy source voltage  $v_{\text{tx}}$  is  fixed, it follows that $p_{\text{tx}}$ given  in (\ref{eq:PT})  strictly increases over $x_n>0$. 

\begin{proposition} \label{Prop:4}
$p_m$,  $\forall m \neq n$,  strictly increases  over  $x_n >0$. However,  $p_n$  first increases  over $0<x_n< \dot{x}_n$, and then decreases over $x_n > \dot{x}_n$, where  
\begin{align} \label{eq:X_Star}
\dot{x}_n=\dfrac{r_n\left(r_{\textnormal{tx}}+\phi_n\right)+w^2 h_n^2}{r_{\textnormal{tx}}+\phi_n},
\end{align}
with  
$\phi_n=w^2 \sum_{k\in {\cal N}\setminus\{n\}} h_k^2(r_k+x_k)^{-1}$.
\end{proposition}
\begin{proof}
Please see Appendix B.
\end{proof}

The above result can be explained as follows.   From  (\ref{eq:ILm}), it follows that  for each receiver $m$, $m \neq n$, its current $|i_m|$ strictly increases over  $x_n>0$.    
The received power  $p_m$ defined in (\ref{eq:PL}) thus strictly increases over $x_n>0$. On the other hand, it follows from (\ref{eq:ILm}) that for receiver $n$,  its current $|i_n|$  strictly decreases  over $x_n>0$. However, from (\ref{eq:PL}), it follows that the decrement in $|i_{n}|^2$ is smaller than the increment of $x_n$ when $0< x_n <\dot{x}_n$; therefore, $p_n$ increases over $x_n$ in this region. The opposite is  true when $x_n > \dot{x}_n$.

\begin{proposition} \label{Prop:5}
If $r_{\textnormal{tx}}+\phi_n-2\varphi_n \le 0$,    $p_\textnormal{sum}$ strictly increases over $x_n>0$, where 
$\varphi_n=w^2\sum_{k\in {\cal N}\setminus\{n\}} h_k^2 x_k (r_k+x_k)^{-2}$; 
otherwise,   $p_\textnormal{sum}$  first increases over   $0<x_n<\ddot{x}_n$,  and then decreases over $x_n>\ddot{x}_n$, where  
\begin{align}
\ddot{x}_n=\dfrac{r_n\left(r_{\textnormal{tx}}+\phi_n\right)+w^2h_n^2+2r_n \varphi_n}{r_{\textnormal{tx}}+\phi_n-2\varphi_n}. 
\end{align}
\end{proposition}
\begin{proof}
Please see Appendix C.
\end{proof}

This result is a consequence of Proposition  \ref{Prop:4}, from which it is known  that  $p_m$'s, $m\neq n$, strictly increase over $x_n>0$, while $p_n$ first increases over $0<x_n<\dot{x}_n$ and then decreases over $x_n> \dot{x}_n$. The sum-power $p_{\text{sum}}$ can thus behave similarly as either $p_m$'s (monotonically increasing) or $p_n$ (initially increasing and then decreasing) over $x_n>0$, depending on the system parameters.

\begin{proposition} \label{Prop:6}
If $\varphi_n-\phi_n-r_{\text{tx}} \ge 0$, $\rho$ strictly increases over $x_n>0$; otherwise, $\rho$ first increases over $0<x_n<\dddot{x}_n$,  and then decreases over $x_n>\dddot{x}$,  where 
\begin{align} \label{eq:xmDag}
\dddot{x}=\dfrac{-r_n\varphi_n-\sqrt{r_n^2\varphi_n^2-\Gamma_n}}{\varphi_n-\phi_n-r_{\text{tx}}},
\end{align}
with $\Gamma_n=(\varphi_n-\phi_n-r_{\text{tx}}) (r_n^2(r_{\text{tx}}+\varphi_n+\phi_n)+r_n w^2h_n^2 )$.
\end{proposition}
\begin{proof}
Please see Appendix D.
\end{proof}

This result is a consequence of Propositions \ref{Prop:3} and \ref{Prop:5}, due to the different behaviors of 
$p_{\text{tx}}$ and $p_{\text{sum}}$ over $x_n>0$.
\subsection{Validation of Analysis}
\label{subsec:NumericalExample-Performnce}
\begin{figure} [t!]
\centering
\includegraphics[width=12cm]{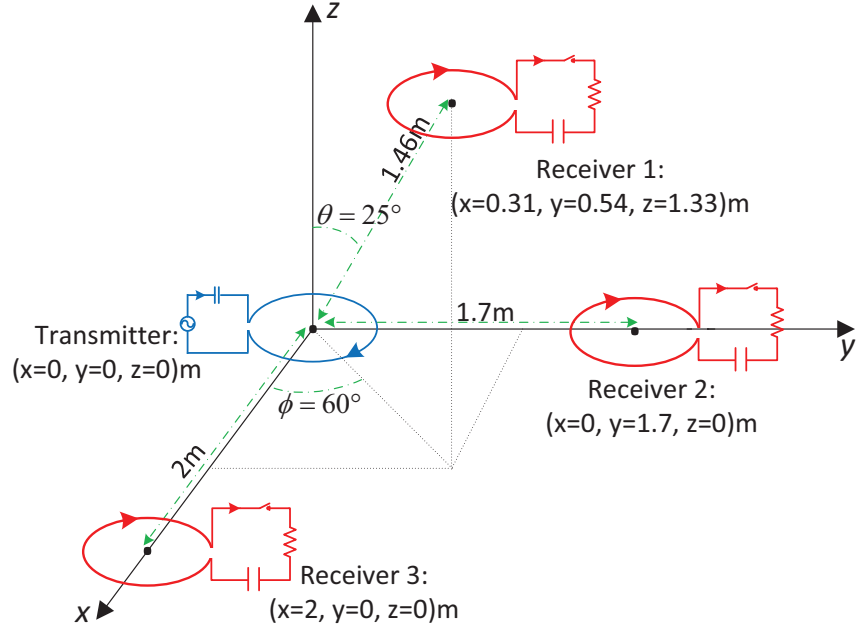} 
\caption{The considered system setup for numerical examples.}
\label{fig:System-Setup} 
 \end{figure}
 
For the purpose of exposition, we consider a point-to-multipoint MRC-WPT system with  $N=3$ receivers, as shown in Fig. \ref{fig:System-Setup}, where  the transmitter and receivers use circular EM coils (see Fig. \ref{fig:Coil-Example} in Appendix \ref{App:Paramters}),   with the physical characteristics given in Table \ref{tab:Tab-1}. Note that the transmitter and both receivers $2$ and $3$ lie in the plane with $z=0$, while receiver $1$ lies in the plane with $z=0.91$ meter (m). Accordingly, the internal resistance and self-inductance of individual EM coils  as well as the mutual inductance among them can be derived (see the details in  Appendix \ref{App:Paramters}), where the obtained values are   given in Table \ref{tab:addlabel}. In this example, although all receivers use EM coils with the same physical characteristics, they are located in different distances from the transmitter.
Specifically, receiver $1$ is closest to the transmitter and  thus has the largest mutual inductance with the transmitter, while receiver $3$ is farthest and has the smallest mutual inductance.
\begin{table}[t!]
	\centering
	\caption{Physical characteristics of EM coils}
	\begin{tabular}{|c|N|N|N|N|N|D|}
		\hline
		\vspace{.1mm}EM Coil &  Inner radius  (cm) &  Outer radius  (cm) & Average radius (cm) & Number of turns & Material of wire& Resistivity of wire ($\mu \Omega$/m) \bigstrut\\ \hline
		\hspace{-1.8mm}Transmitter &     $19.9$   &     $20.1$ & $20$ &  $200$     &    Copper   &  $0.0168$ \bigstrut\\ \hline
		\hspace{-1mm}Receiver $1$ &       $4.95$  &   $5.05$   & $5$ &    $10$  &   Copper    & $0.0168$   \bigstrut\\ \hline
		\hspace{-1mm}Receiver $2$ &   $4.95$   &    $5.05$  & $5$ & $10$      &    Copper  &  $0.0168$  \bigstrut\\ \hline
		\hspace{-1mm}Receiver $3$ &   $4.95$   &     $5.05$ & $5$ &  $10$     &    Copper   &  $0.0168$ \bigstrut\\ \hline
	\end{tabular}
	\label{tab:Tab-1}
\end{table}
\begin{table}[t!]
	\vspace{2mm}
	\centering
	\caption{Electrical characteristics of EM coils}
	\begin{tabular}{|c|S|S|S|S|}
		\hline
		\vspace{.1mm}EM Coil &  Internal resistance $r_{\text{tx}}$/$r_n$ ($\Omega$) & Self-inductance $l_{\text{tx}}$/$l_n$ (mH) & Mutual inductance $h_n$ ($\mu$H)  \bigstrut\\ \hline
		\hspace{-1.8mm}Transmitter & $1.3440$ &  $54.0630$  &  --  \bigstrut\\ \hline
		\hspace{-1mm}Receiver $1$ &  $0.0672$ &  $0.0294$ & $\hspace{-1.1mm}-0.0921$ \bigstrut\\ \hline
		\hspace{-1mm}Receiver $2$ & $0.0672$  &  $0.0294$   & $~0.0402$  \bigstrut\\ \hline
		\hspace{-1mm}Receiver $3$ &  $0.0672$ & $0.0294$ &  $~0.0245$  \bigstrut\\ \hline
	\end{tabular}
	\label{tab:addlabel}
\end{table}
We set $v_{\text{tx}}=20 \sqrt{2}$V, and  $w=42.6 \times10^6$rad/s (i.e., $6.78 $MHz), as suggested in the A4WP specification \cite{Nadakuduti}. For this example, we fix $x_{2}=x_{3}=2.5\Omega$. 

First, we plot   $p_{\text{tx}}$,  $p_{n}$'s, $\forall n \in \cal N$, and $p_{\text{sum}}$ versus the resistance of  load $1$, i.e., $x_{1}$,  in Fig.  \ref{fig:Eff-Pl-versus-rL1}. 
It is observed that  $p_{\text{tx}}$, $p_2$, $p_3$ and $p_{\text{sum}}$ all increase over $x_1>0$; however,  $p_1$  initially increases over $0<x_1<\dot{x}_1=5.35\Omega$ and then  declines  over $x_1>5.35\Omega$.  
Note that in this example,   the condition  $r_{\textnormal{tx}}+\phi_n-2\varphi_n<0$ holds in Proposition \ref{Prop:5}.  The obtained results are all consistent with our analysis in Section III-A.
Besides, we observe that varying $x_{1}$ not only changes $p_{1}$, but also  the power delivered to other loads.   For instance,   receiver $1$ can help  receivers $2$ and $3$, which are farther away from the transmitter,  to receive more power by increasing its load resistance $x_1$. This is a useful mechanism that will be utilized later in this paper to mitigate the near-far issue.

Second, in Fig. \ref{fig:Eff-versus-rL1}, we plot the sum-power transfer efficiency $\rho$ as a function of $x_1$.  It is  observed that  $\rho$ follows a single-peak pattern over $x_n>0$, i.e., it  first increases over $0<x_1<\dddot{x}_1=0.95\Omega$, and then smoothly declines   over $x_1>0.95\Omega$. 
This result can be verified from  Proposition \ref{Prop:6} by considering the fact that in this example,   the condition $\varphi_n-\phi_n-r_{\text{tx}} < 0$ holds.  Note that when $x_1 \rightarrow \infty$, it follows from (\ref{eq:IT1m}) and (\ref{eq:PT}) that $i_1 \rightarrow 0$ and $p_1\rightarrow 0$. This is equivalent to disconnecting load $1$ from receiver $1$, i.e., setting $s_1=0$. As a result, the efficiency converges when  $x_1 \rightarrow \infty$, while the converged value depends on the parameters of the transmitter and the other two receivers.

Third, we set $x_1=x_2=x_3=2.5\Omega$, and plot the  power received by the three loads versus $w$ in Fig. \ref{fig:Power-vs-w}. It is observed that  $p_1$, $p_2$, and $p_3$  reach their individual peaks all at  $w=\dot{w}=17.97\times 10^6$rad/sec, which is in accordance to Remark \ref{Remark:w}.  It is worth noting that although $w$ and $x_n$'s can be jointly designed to achieve better performance over the case of optimizing $x_n$'s only with $w$ being fixed, this  problem is challenging to solve and thus is left as our future work.
\begin{figure}[t!]
	\centering
	\includegraphics[width=12cm]{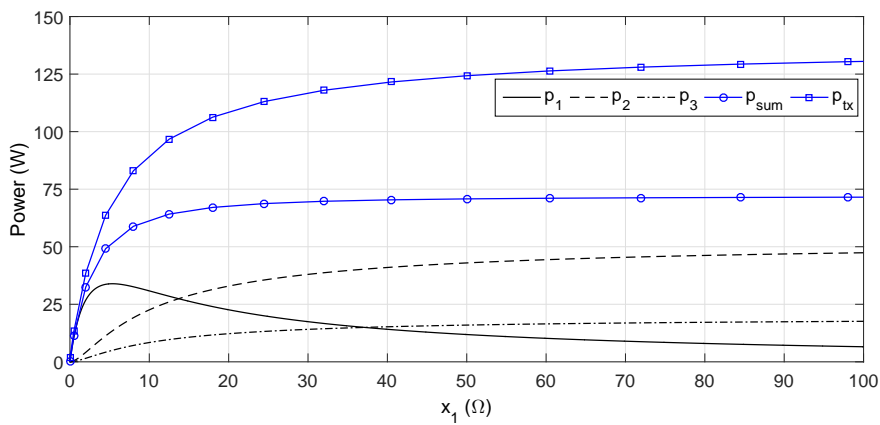}
	\caption{Input and output power versus $x_1$.} 
	\label{fig:Eff-Pl-versus-rL1}
\end{figure}
\begin{figure}[t!]
	\centering
	\includegraphics[width=12cm]{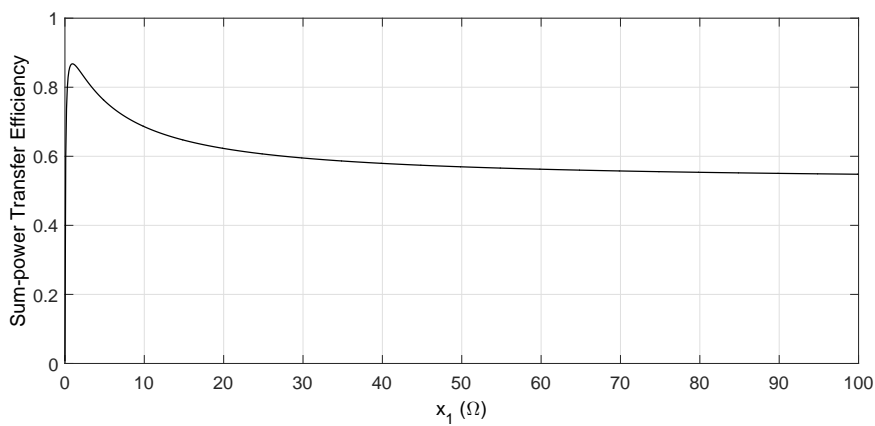}  
	\caption{The sum-power transfer efficiency  $\rho$ versus $x_1$.} 
	\label{fig:Eff-versus-rL1}
\end{figure} 
\begin{figure}[t!]
	\centering
	\includegraphics[width=12cm]{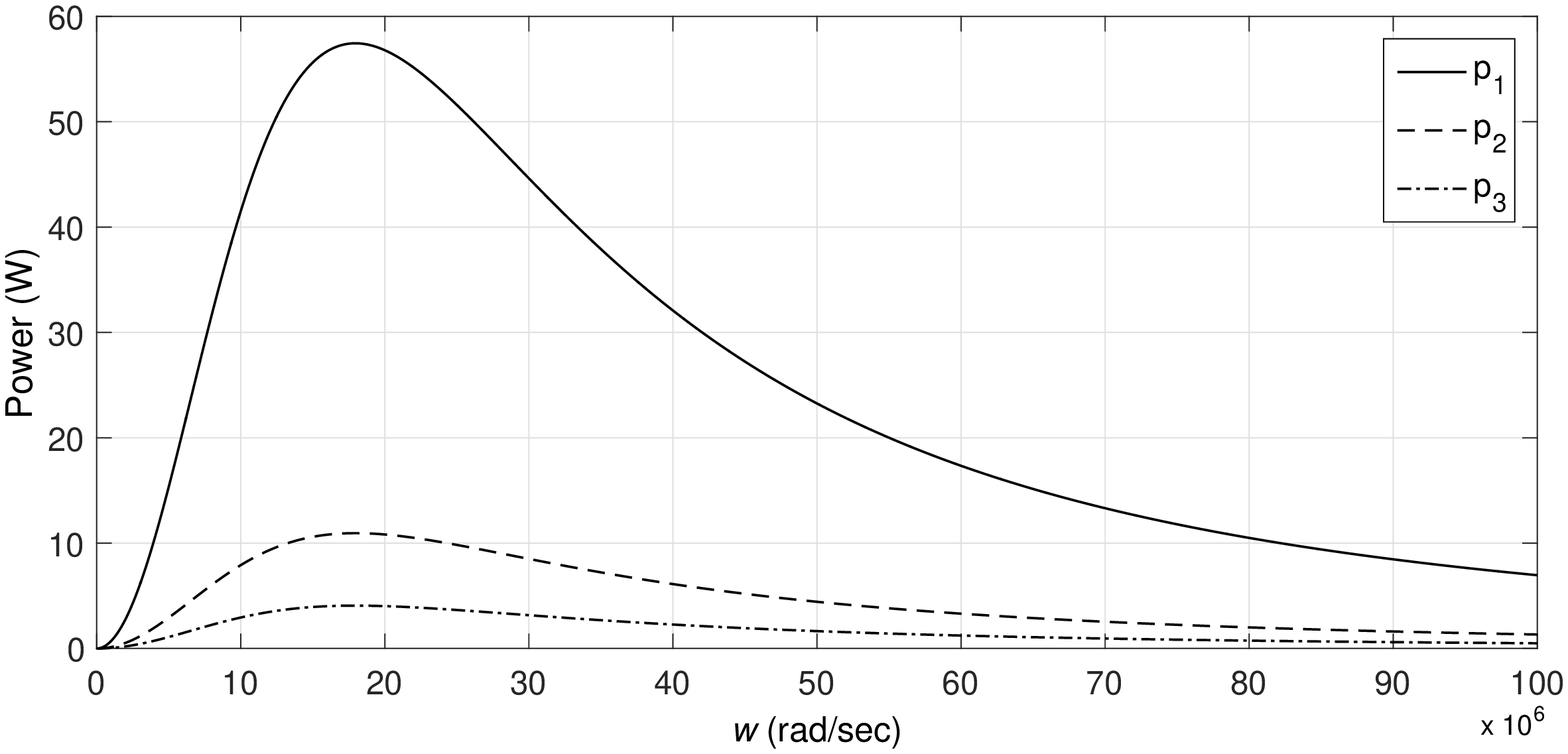}  
	\caption{The output power versus $w$.} 
	\label{fig:Power-vs-w}
\end{figure}
\section{Multiuser Charging Control Optimization} \label{Sec:Cent-Cha-Cont}
In this section, we  optimize the receivers' load  resistance  $x_{n}$'s  to minimize  the transmitter power $p_{\text{tx}}$ subject  to the given load constraints, by assuming that  $s_n=1$, $\forall n \in \cal N$, i.e., all the receives are connected to their loads. 
First, we consider the case with a  central controller at the transmitter, which  has the full knowledge of all receivers, including their circuit   parameters as well as their load  requirements,  to implement  centralized charging control. 
We then devise a distributed  charging algorithm for the receivers to independently adjust their load resistance iteratively, for the ease of practical implementation.
Last, we compare the performance of the two algorithms under a practical system setup.   
\subsection{Problem Formulation}
We  assume that in practice the resistance of each load $n$  can be adjusted over a given range $\underline{x}_n \le x_n \le \overline{x}_n$, where $\underline{x}_n >0$ ($\overline{x}_n \ge\underline{x}_n$) is the  lower (upper) limit  of the resistance. 
It is also assumed that the power delivered to each  load $n$ needs to be  higher than a given    minimum threshold  $\underline{p}_{n}>0$  to guarantee its quality of service.  
Next, we formulate the  optimization problem (P1) to minimize the  transmitter power $p_{\text{tx}}$ subject to the given load constraints of all receivers as follows. 
\begin{align} 
\mathrm{(P1)}: \hspace{-1mm}
\mathop{\mathtt{min}}_{\{ \underline{x}_n\le x_{n}\le \overline{x}_n\}_{n\in \cal N}}
&~\dfrac{|v_\text{tx}|^2}{2}\dfrac{1}{r_{\text{tx}}+ w^2 \sum_{k \in \cal N} h_{k}^2\left(r_{k}+x_{k}\right)^{-1}} \\
\mathtt{s.t.} 
&~\dfrac{|v_\text{tx}|^2}{2}\dfrac{w^2 h_n^2 x_{n} \left(r_{n}+x_{n}\right)^{-2}}{\left(r_{\text{tx}}+ w^2 \sum_{k \in \cal N} h_{k}^2\left(r_{k}+x_{k}\right)^{-1} \right)^2} \ge \underline{p}_n, \forall n \in \cal N. \label{eq:Const:p1-1}
\end{align} 
Although (P1) is  non-convex,  we propose a \textit{centralized}  algorithm to solve it optimally in the next subsection. 
\subsection{Centralized Algorithm}
First, based on (P1),  we  formulate the  maximization problem (P2), where its objective function is the  inverse of that of (P1) but with the same constraints as   (P1).  
\begin{align} 
\mathrm{(P2)}: \hspace{-1mm}
\mathop{\mathtt{max}}_{\{ \underline{x}_n\le x_{n}\le \overline{x}_n\}_{n\in \cal N}}
&~\dfrac{2}{|v_\text{tx}|^2} \left(r_{\text{tx}}+ w^2 \sum_{k \in \cal N} h_{k}^2\left(r_{k}+x_{k}\right)^{-1} \right) \\
\mathtt{s.t.} 
&~\dfrac{|v_\text{tx}|^2}{2}\dfrac{w^2 h_n^2 x_{n} \left(r_{n}+x_{n}\right)^{-2}}{\left(r_{\text{tx}}+ w^2 \sum_{k \in \cal N} h_{k}^2\left(r_{k}+x_{k}\right)^{-1} \right)^2} \ge \underline{p}_n,~ \forall n \in \cal N. \label{eq:Const:p2-1}
\end{align} 
It can be verified that the optimal solution to (P2)  also solves (P1); as a result, we can equivalently solve (P2) to derive the optimal solution to (P1). Although (P2) is still non-convex, we can re-formulate it as a convex problem by applying   change of variables. Specifically, we define a new set of variables as $y_n=1/(r_n+x_n)$, $\forall n \in \cal N$. Since $\underline{x}_n\le x_{n}\le \overline{x}_n$, it  follows that $\underline{y}_n\le y_n\le \overline{y}_n$, where $\underline{y}_n=1/(r_n+\overline{x}_n)$ and $\overline{y}_n=1/(r_n+\underline{x}_n)$.  Accordingly, we  rewrite (P2) as  (P3). 
\begin{align} 
\mathrm{(P3)}: \hspace{-1mm} 
\mathop{\mathtt{max}}_{\{ \underline{y}_n\le y_{n}\le \overline{y}_n\}_{n\in \cal N}}
&~\dfrac{2}{|v_\text{tx}|^2}{\left(r_{\text{tx}}+ w^2 \sum_{k \in \cal N} h_{k}^2 \hspace{0.5mm}y_k \right)} \\
\mathtt{s.t.} 
&~\dfrac{|v_\text{tx}|^2}{2} w^2 h_n^2 \left(r_n y_n^2- y_n\right)+\underline{p}_n \left(r_{\text{tx}}+ w^2 \sum_{k \in \cal N} h_{k}^2 \hspace{0.5mm} y_k \right)^2 \le 0,~ \forall n \in \cal N. \label{eq:Const:p3-1}
\end{align} 
Note that (P3) is a  convex optimization  problem, with a linear objective function and linear/quadratic inequality constraints over  $y_n$'s. As a result,  (P3) can be efficiently solved  using the existing software, e.g., CVX \cite{Boyd3}.  
Let $(y_1^*,\ldots,y_N^*)$ denote the optimal solution to (P3).  The optimal solution to (P2) is thus obtained by a change of variable as $x_n^*=1/y_n^* - r_n$, $\forall n \in \cal N$. 
The obtained $(x_1^*,\ldots,x_N^*)$ also solves (P1).
The  centralized algorithm  to solve (P1) is summarized  in Table \ref{Tabel:Centr}, denoted as Algorithm $1$.  
Since the feasibility of convex problem (P3) can be efficiently checked, in the rest of this paper, we assume that (P1), or equivalently (P3), is  feasible  without loss of generality.\hspace{-1mm}
\begin{table}[t!]
\begin{center} 
\caption{Centralized algorithm  for (P1).} \scriptsize{
 \hrule\vspace{0.1cm} 
\textbf{Algorithm $1$}
\hrule 
\begin{itemize}
\item[a)] For each receiver $n$, $\forall n \in \cal N$, given  $\underline{x}_n>0$ and  $\overline{x}_n> \underline{x}_n$, compute  $\underline{y}_n=1/(r_n+\overline{x}_n)$ and $\overline{y}_n=1/(r_n+\underline{x}_n)$. Accordingly, formulate the problem (P3).
\item[b)] {\bf If} (P3) is feasible, {\bf then} save its optimal solution as $(y_1^*,\ldots,y_N^*)$. Set $x_n^*=1/y_n^* - r_n$, $\forall n \in \cal N$. Return $(x_1^*,\ldots,x_N^*)$ as the optimal solution to (P1). 
\item[c)]{\bf If} (P3) is infeasible, {\bf then} it follows that there is no feasible solution to (P1) and thus the algorithm terminates.
\end{itemize} 
\hrule \label{Tabel:Centr} }
\end{center} 
\end{table}
\subsection{Distributed  Algorithm} \label{Sec:dist-Alg} 
In this subsection, we present  an alternative \textit{distributed} algorithm for  (P1), for the case without a  central controller  installed in the system.   In this algorithm, each   receiver adjusts its load  resistance independently according to its local information and a  one-bit feedback received from each of the other receivers indicating  whether the corresponding load  constraint is  satisfied or not. 
We denote the feedback from each receiver $n$ which is broadcasted to all other receivers  as $FB_n \in \{0,1\}$,  where $FB_n=1$ ($FB_n=0$) indicates that its load constraint is (not)  satisfied. 

In Section \ref{sec:Performance}, we show  that   the power delivered to each load  $n$,  $p_n$,  has two properties that can be exploited to adjust  $x_n$. 
First, $p_n$ strictly increases over $x_m>0$, $\forall m \neq n$, which means that other receivers can help boost $p_n$ by increasing their individual load resistance.
Second, $p_n$ has a single peak at $x_n=\dot{x}_n$, assuming that  all other load  resistance is fixed. 
Thus, over $0<x_n<\dot{x}_n$,  receiver $n$  can increase $p_n$  by increasing  $x_n$; similarly, for $x_n>\dot{x}_n$, it can increase $p_n$ by reducing $x_n$. 
Although  receiver $n$  cannot  compute  $\dot{x}_n$ from (\ref{eq:X_Star}) directly due to its incomplete information on  other receivers, it can  test whether  $0< x_n< \dot{x}_n$, $x_n=\dot{x}_n$, or $x_n>\dot{x}_n$ as follows. 
Let $p_n({x_n^+})$, $p_n({x_n})$, and $p_n(x_n^-)$ denote the power received by  load $n$ when its resistance is  set as $x_n+\Delta x$, $x_n$, and $x_n-\Delta x$, respectively, where $\Delta x>0$ is a small step size.  Assuming all the other load resistance is fixed,  receiver $n$ can make the following decision:\\
$\bullet$ If $p_n({x_n^+})>p_n({x_n})$ and $p_n({x_n^-})<p_n({x_n})$, then   $0<x_n<\dot{x}_n$; \\
$\bullet$ If $p_n({x_n^+})<p_n({x_n})$ and $p_n({x_n^-})<p_n({x_n})$, then  $x_n= \dot{x}_n$;\footnote{More precisely,   in this case we have $\dot{x}_n-\Delta x\le x_n \le \dot{x}_n+\Delta x$.} \\
$\bullet$ If $p_n({x_n^+})<p_n({x_n})$ and $p_n({x_n^-})>p_n({x_n})$, then  $x_n>\dot{x}_n$.

Now, we present the distributed algorithm in detail. 
The algorithm is implemented in an iterative manner, say, starting from receiver 1, where in each iteration, only one receiver $n$ adjusts its load resistance, while all the other receivers just broadcast their individual one-bit feedback $FB_m$, $m\neq n$, at the beginning of each iteration. 
We initialize  $x_n=\min\{\hspace{.7mm}\max\{({r_n r_{\textnormal{tx}}+w^2 h_n^2})/r_{\textnormal{tx}},\hspace{.7mm}\underline{x}_n\},\hspace{.7mm} \overline{x}_n\}$, $\forall n \in \cal N$, where $({r_n r_{\textnormal{tx}}+w^2 h_n^2})/r_{\textnormal{tx}}$ is obtained from (\ref{eq:X_Star}) by setting $\phi_n=0$, i.e., assuming that all other receivers have their loads disconnected.\footnote{This requires a protocol design so that when each new receiver is added in the system, its mutual inductance $h_n$ is measured and $x_n$ is accordingly computed and initially set at the receiver.} This is a reasonable starting point,  under which the power delivered to each receiver is maximized (see Proposition \ref{Prop:4}). Then, as the algorithm proceeds, all  receivers can gradually adjust their load resistance to help reduce the transmit power  while meeting the minimum power constraints of their individual loads.    
Specifically, at each iteration for receiver $n$, if   $p_n < \underline{p}_n$, then it will adjust $x_n$ to increase $p_n$. 
To find the direction for the update, the receiver needs to check for its current  $x_n$ whether  $0< x_n< \dot{x}_n$,  $x_n=\dot{x}_n$, or $x_n> \dot{x}_n$ holds, using the method aforementioned. On the other hand,  if $p_n > \underline{p}_n$, receiver $n$ can increase $x_n$  to help increase the power delivered to other loads when there exists any $m\neq n$ such that $FB_m=0$ is received;  or it can decrease $x_n$  to help reduce the transmitter power when $FB_m=1$, $\forall m \neq n$. 
In summary, we design the following protocol (with five cases)  for  receiver $n$ to update $x_n$.  \\
Case 1: If $p_n < \underline{p}_n$ and $0< x_n< \dot{x}_n$, set  $x_n \gets \min\{\overline{x}_n, x_n+ \Delta x \}$. \\
Case 2: If $p_n < \underline{p}_n$ and $x_n> \dot{x}_n$, set  $x_n \gets \max \{\underline{x}_n, x_n- \Delta x\}$. \\
Case 3: If $p_n > \underline{p}_n$, $x_n\neq {\dot x}_n$, and  $\exists m \neq n$, $FB_m=0$, set $x_n \hspace{.35mm} \gets \hspace{.35mm}\min  \{\overline{x}_n, x_n+\Delta x\}$.  \\
Case 4: If $p_n > \underline{p}_n$, $x_n\neq {\dot x}_n$, and   $FB_m=1$, $\forall m \neq n$, set $x_n \gets \max \{\underline{x}_n, x_n-\Delta x\}$.  \\
Case 5: Otherwise, no update occurs. 

We set a maximum number of iterations, denoted by  $itr_{\max}$, after which the algorithm will terminate.  The above distributed algorithm  is summarized in Table \ref{Tabel:Dist}, denoted as Algorithm $2$. It is worth noting that due to the simplicity of Algorithm $2$  as well as its distributed nature, this algorithm  may not converge to the optimal solution to (P1) in general, or may even fail to converge to a feasible solution to (P1) in certain cases, as will be shown by the numerical example presented next. 

\begin{table}[t!]
\begin{center} 
\caption{Distributed algorithm for (P1).} \scriptsize{
 \hrule\vspace{0.1cm} 
\textbf{Algorithm $2$}
\hrule 
\begin{itemize}
\item[a)]  Initialize $itr=1$ and $itr_{\max}>1$. Each receiver $n$ sets  $x_n=\min\{\hspace{.7mm}\max\{({r_n r_{\textnormal{tx}}+w^2h_n^2})/r_{\textnormal{tx}},\hspace{.7mm}\underline{x}_n\},\hspace{.7mm} \overline{x}_n\}$. 
\item[b)]  {\bf Repeat} from receiver $n=1$ to $n=N$:
\begin{itemize}
\item[$\bullet$] Receiver $n$  collects $FB_m$ from all other receivers $m\neq n$. 
\item[$\bullet$] Receiver $n$ updates its load resistance $x_n$ according to Cases 1--5. 
\item[$\bullet$] {\bf If} \hspace{1mm}$itr=itr_{\max}$, {\bf then} quit the loop  and the algorithm terminates.
\item[$\bullet$] Set $itr=itr+1$. 
\end{itemize}
\end{itemize}
\hrule  \label{Tabel:Dist} }
\end{center}
\end{table}  
\subsection{Performance Comparison}
We consider the  same system setup as that in Section  \ref{subsec:NumericalExample-Performnce}. We set $\underline{x}_n=1\Omega$ and $\overline{x}_n=100\Omega$, $\forall n\in \cal N$. 
We also set $\underline{p}_1=\underline{p}_2=17.5$W, but vary $p_3$ over  $0<\underline{p}_3 \le 37.95$W, where (P1) can be verified to be feasible in this specific region. 
For Algorithm $2$,  we use $\Delta x=10^{-3}$ and ${itr}_{\max}=3\times10^{5}$, with  ${itr}_{\max}> \sum_{k=1}^{N} (\overline{x}_k-\underline{x}_k)/\Delta x$, which is sufficiently large such that each receiver $n$ can search for its load resistance $x_n$ over the whole range of $[\underline{x}_n,~\overline{x}_n]$ before the algorithm terminates.  

Fig. \ref{fig:Simul-alg} compares the transmitter power $p_{\text{tx}}$ obtained  by both Algorithms $1$ and $2$ versus  $\underline{p}_3$. In this example,  Algorithm $2$  converges to a feasible solution to (P1) only over $0<\underline{p}_3\le33.75$W, while it yields an infeasible solution to (P1) if  $\underline{p}_3> 33.75$W. Particularly, with  $\underline{p}_3> 33.75$W, the load resistance  $(x_1,x_2,x_3)$ obtained via Algorithm $2$ can satisfy the power constraints of loads $1$ and $2$, but not that of load $3$.  Moreover, it is observed that with $0<\underline{p}_3\le33.75$W, Algorithm $2$ achieves almost the same minimum  $p_{\text{tx}}$ as that by  Algorithm $1$, which solves (P1) optimally. Notice that when  $\underline{p}_3> 33.75$W,  the obtained   $p_{\text{tx}}$ via Algorithm $2$ is lower than that of Algorithm $1$. However, this result is not meaningful as the solution by Algorithm $2$ in this case is not feasible.

Fig. \ref{fig:Simul-algConver} shows the convergence of Algorithm 2  under the above system setup with $\underline{p}_3=30$W. It is observed that  although  $p_1>\underline{p}_1$ and $p_2>\underline{p}_2$ at the first iteration, we have $p_3<\underline{p}_3$. As a result, both receivers $1$ and $2$ help receiver $3$ (which is most far-way from the transmitter)  for receiving more power by lowering their received power levels via increasing their individual load resistance. It is also observed that this algorithm takes around $0.4\times 10^5$ iterations to converge, since we use  $\Delta x =10^{-3}$ in the algorithm for updating $x_n$'s, which is a small step size to ensure smooth convergence. In practice, larger step size can be used to speed up the algorithm but at the cost of certain  performance loss.   

\begin{figure}
\centering
\includegraphics[width=12cm]{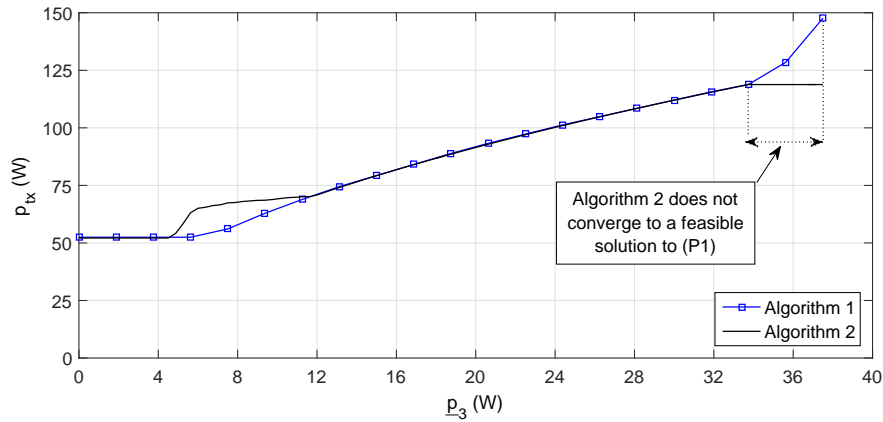} 
\caption{Performance comparison between centralized control (Algorithm 1) versus distributed control (Algorithm 2).} 
\label{fig:Simul-alg}
\end{figure}
\begin{figure}
\centering
\includegraphics[width=12cm]{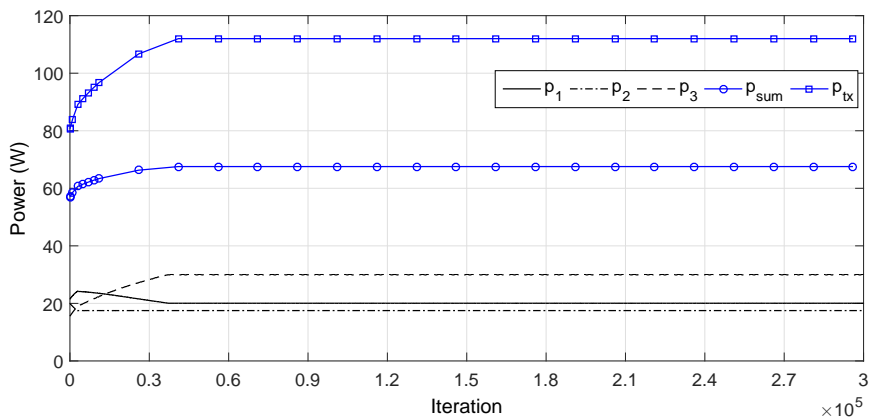}
\caption{Convergence performance of Algorithm $2$.} 
\label{fig:Simul-algConver}
\end{figure}
\section{Power Region Characterization}
In this section, we  characterize the achievable power region for the receiver loads   without versus with time sharing.  First, we propose a time-sharing scheme  to schedule multiuser power transfer  by  connecting/disconnecting loads to/from  their  receivers over time. 
We then propose a centralized algorithm to jointly optimize the time allocation and load resistance of receivers for time sharing based power transmission, by extending  that for (P1) in Section \ref{Sec:Cent-Cha-Cont} for the case without time sharing.  Last, numerical examples are  provided to  compare the power-region  performance of the multiuser  MRC-WPT system without versus with time sharing.
\subsection{Multiuser Power Transfer with Time Sharing}
\begin{figure}[t!] 
\centering 
\includegraphics[width=10.2cm]{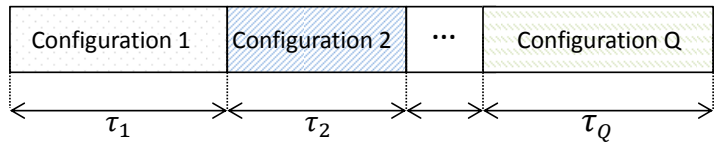} 
\caption{Time-sharing based multiuser power transfer.}  
\label{fig:time-split} 
\end{figure}  
As shown in Fig. \ref{fig:time-split}, there are in general $Q=2^N-1$ time-sharing  configurations for transferring power to $N$ receiver loads depending on the state of each receiver's switch. We index these configurations by $q$, $q\in {\cal Q}=\{1,\ldots,Q\}$. 
Specifically, let  ${\cal S}=\{(s_1,\ldots,s_N)~|~ s_n\in \{0,1\},~ \forall n \in {\cal N}\}$ denote  the set  consisting of all possible states of receiver switches. Without loss of generality, we remove the trivial case that all switches are open, i.e., $s_n=0$, $\forall n \in \cal N$, from $\cal S$ by setting ${\cal S}\leftarrow {\cal S}\setminus\{(0,...,0)\}$. As a result, we have the cardinality of $\cal S$ as  $|{\cal S}|=Q$. For convenience, we  assign  an one-to-one mapping between the elements in the two sets $\cal{Q}$ and   $\cal S$, i.e.,  we assign each configuration  $q \in \cal Q$ to  one  switch state  $(s_1,\ldots,s_N) \in \cal S$. By default, we assign configuration $q=1$ to  $(s_1,\ldots,s_N)=(1,\ldots,1)$, i.e., all switches are closed.

Let $\tau>0$ denote the total time available for power transmission.  
Let the time allocated for the power transmission under  configuration $q$ be denoted by $\tau_q$, with $0\le \tau_q \le \tau$.
We thus have  $\sum_{q \in \cal Q} \tau_q \le \tau$, where the strict inequality occurs when the required energy, $\overline{p}_n\tau$, $n=1,...,N$, at all loads are satisfied by the end of $Q$-slot transmissions, where the voltage source at the transmitter can be switched off for the remaining time ($\tau-\sum_{q\in{\cal Q}}\tau_q)>0$ to save energy.\footnote{Note that disconnecting all receivers from their loads, i.e. setting $s_n=0$, $\forall n \in \cal N$, cannot achieve this goal. This is due to the fact that the ohmic resistance of the transmitter circuit still consumes power as long as the transmitter voltage source is on.} 
Over all configurations, the average  transmitter power and the average power delivered to each load  $n$ can be obtained from (\ref{eq:PT}) and (\ref{eq:PL}), respectively,  as follows:
\begin{align}
&p_{\text{tx}}\hspace{-.5mm}=\hspace{-2mm}\sum_{q \in {\cal Q}}\hspace{-1mm}\dfrac{|v_{\text{tx}}|^2}{2\tau} \dfrac{1}{r_{\text{tx}}+ w^2 \sum_{k \in {\cal N}_q} h_{k}^2\left(r_{k}+x_{k,q}\right)^{-1}}\tau_q, \label{eq:P-TX-TS}\\
&p_n\hspace{-.5mm}=\hspace{-2mm}\sum_{q \in {\cal Q}_n}\hspace{-1mm} \dfrac{|v_{\text{tx}}|^2}{2\tau} \dfrac{w^2 h_n^2 x_{n,q} \left(r_{n}+x_{n,q}\right)^{-2}}{\left(r_{\text{tx}}+ w^2 \sum_{k \in {\cal N}_q} h_{k}^2\left(r_{k}+x_{k,q}\right)^{-1}\right)^2} \tau_q,\hspace{-1.5mm} \label{eq:Pn-TS}
\end{align} 
where $x_{n,q}$ is the resistance value of load $n$ under configuration $q$.  Furthermore, ${\cal N}_q  \subseteq \cal N$  denotes the subset of receivers with their loads connected under configuration $q$, while  ${\cal Q}_n \subseteq \cal Q$ denotes the subset of configurations under which receiver $n$ has its load connected. 
Compared to the previous case without time sharing, the time allocation $\tau_q$'s, $\forall q \in \cal Q$,  can provide  extra  degrees of freedom for  performance optimization.
\subsection{Power Region Definition}
The \textit{power region} is defined as the set of power-tuples achievable for all loads with a given transmission time $\tau$ subject to their adjustable resistance values.  Specifically, the power region for the case without time sharing (i.e., all loads receive power concurrently)  is defined as 
\begin{equation} \label{eq:Reg_WiOu}
{\cal R}_{\text{without-TS}}=  \bigcup_{\{\underline{x}_n \le x_n\le \overline{x}_n\}_{n \in \cal N}}  (p_1,\ldots,p_N),
\end{equation}
with $p_n$'s, $\forall n \in \cal N$, are given in (\ref{eq:PL}).  Similarly, the power region for the case with time sharing is   defined as
\begin{equation} \label{eq:Reg_With}
{\cal R}_{\text{with-TS}}= \bigcup_{\substack{ \{0\le\tau_q \le \tau,\}_{q \in \cal Q},~ \sum_{q \in \cal Q} \tau_q \le \tau \\ \{\underline{x}_n \le x_{n,q}\le \overline{x}_n\}_{n \in {\cal N},q \in \cal Q}} } (p_1,\ldots,p_N),
\end{equation}
with $p_n$'s, $\forall n \in \cal N$, are given in (\ref{eq:Pn-TS}). It is evident  that the power region with time sharing is no smaller than that without time sharing in general, i.e., ${\cal R}_{\text{without-TS}} \subseteq {\cal R}_{\text{with-TS}}$, since by simply setting $\tau_{1}=\tau$ and $\tau_q=0$, $\forall q \neq 1$, we have ${\cal R}_{\text{without-TS}} ={\cal R}_{\text{with-TS}}$.  
\subsection{Centralized Algorithm  with Time Sharing: Revised }
In this subsection, we extend the centralized algorithm in Section \ref{sec:Performance} without time sharing  to the case with time sharing by jointly optimizing the time allocation  and load resistance of all receivers to minimize the transmitter power subject to the given load (average received power)  constraints.  Hence, we consider  problem (P4) as follows. 
\begin{align} 
\mathrm{(P4)}:\hspace{-1mm} 
\mathop{\mathtt{min}}_{\substack{ \{0\le \tau_q \le \tau\}_{q\in \cal Q},~ \{\underline{x}_n \le x_{n,q} \le \overline{x}_n\}_{n \in {\cal N}, q \in {\cal Q}} }} 
&~\sum_{q \in {\cal Q}} \dfrac{|v_{\text{tx}}|^2}{2\tau} \dfrac{1}{\hspace{2mm}r_{\text{tx}}+ w^2 \sum_{k \in {\cal N}_q} h_{k}^2\left(r_{k}+x_{k,q}\right)^{-1}\hspace{2mm}} \tau_q \label{eq:p4-object}
\\
\mathtt{s.t.} 
&~\hspace{-1.0mm}\sum_{q \in {\cal Q}_n}  \dfrac{|v_{\text{tx}}|^2}{2\tau} \dfrac{w^2 h_n^2 x_{n,q} (r_{n}+x_{n,q})^{-2}}{\left(r_{\text{tx}}+ w^2 \sum_{k \in {\cal N}_q} h_{k}^2\left(r_{k}+x_{k,q}\right)^{-1} \right)^2}  \tau_q\ge \underline{p}_n, ~\forall n \in \cal N, \label{eq:Const1-p4} \\
&~\sum_{q \in \cal Q} \tau_q\le  \tau. \label{eq:Const2-p4}
\end{align} 
Although (P4) is non-convex, we can apply the technique of \textit{alternating optimization}  to solve it sub-optimally in general, as discussed below. Since (P1) is assumed feasible, the feasibility of (P4) is ensured due to the fact that  ${\cal R}_{\text{without-TS}} \subseteq {\cal R}_{\text{with-TS}}$. 

Initialize  $\tau_{1}=\tau$ and  $x_{n,1}=x_{n}^{*}$, $\forall n \in \cal N$, where $(x_1^*,\ldots,x_N^*)$ denotes the  optimal solution to (P1) for the case without time sharing.  Moreover, initialize $\tau_q=0$ and $x_{n,q}=\min\{\hspace{.7mm}\max\{\big({r_n r_{\textnormal{tx}}+w^2 h_n^2}\big)/r_{\textnormal{tx}},\hspace{.7mm} \underline{x}_n \}, \hspace{.7mm} \overline{x}_n\}$, $\forall n \in \cal N$, $\forall q \neq 1$. 
At each iteration $itr$, $itr=1,2,\ldots$, we  design $\tau_q$'s and $x_{n,q}$'s alternatively according to the following procedure. 
First, we solve (P4) over $\tau_q$'s, $\forall q\in \cal Q$,   while the rest of variables are all fixed. 
The resulting  problem is a linear programming (LP) which can be efficiently solved using the existing    software,  e.g. CVX \cite{Boyd3}.  We  update $\tau_q$'s as the obtained solution. 
Next, we optimize the load resistance for different configurations  sequentially, e.g., starting from  configuration $1$ to $Q$. For each configuration $q$, we solve (P4) over $x_{n,q}$'s, $\forall n\in \cal N$,  with the rest of variables all being fixed. We thus consider the optimization problem (P4$-$q) as follows.
\begin{align} 
\mathrm{(P4-q)}:\hspace{-1mm}\mathop{\mathtt{min}}_{\{\underline{x}_n \le x_{n,q} \le \overline{x}_n\}_{n \in {\cal N}} } 
&~\dfrac{|v_{\text{tx}}|^2}{2\tau} \dfrac{1}{\hspace{2mm}r_{\text{tx}}+ w^2 \sum_{k \in {\cal N}_q} h_{k}^2(r_{k}+x_{k,q})^{-1}\hspace{2mm}} \tau_q \label{eq:p4_q-object}
\\
\mathtt{s.t.} 
&~\dfrac{|v_{\text{tx}}|^2}{2\tau} \dfrac{w^2 h_n^2 x_{n,q} \left(r_{n}+x_{n,q}\right)^{-2}}{\left(r_{\text{tx}}+ w^2 \sum_{k \in {\cal N}_q} h_{k}^2(r_{k}+x_{k,q})^{-1} \right)^2}  \tau_q\ge \underline{p}_n-p_{n,-q}, ~\forall n \in \cal N. \label{eq:Const1-q-p4}
\end{align}  
where  $p_{n,-q}=\sum_{m \in {\cal Q}_n\setminus\{q\}} |v_{\text{tx}}|^2  w^2 h_n^2 x_{n,m} (r_{n}+x_{n,m})^{-2}\tau_m /(2\tau(r_{\text{tx}}+ w^2 \sum_{k \in {\cal N}_m} h_{k}^2(r_{k}+x_{k,m})^{-1} )^2)$.   
For each receiver $n$, its load power constraint given in (\ref{eq:Const1-p4}) is re-expressed in  (\ref{eq:Const1-q-p4}), where all power terms that do not involve $x_{n,q}$, $\forall n\in \cal N$, are moved to the right-hand side (RHS) of the inequality, denoted by $p_{n,-q}$, which is treated as constant in (P4$-$q). From (\ref{eq:Pn-TS}), it follows that $p_{n,-q}$ denotes the average power delivered to load $n$ under all other  configurations,  $ m\in {\cal Q}_n\setminus\{q\}$.  
Problem  (P4$-$q)  has the same structure as (P1); as a result, we can solve it using an algorithm similar to Algorithm $1$.
We then update   $x_{n,q}$'s, $\forall n \in \cal N$, as the obtained   solution to (P4$-$q). 
At the end of each iteration $itr$, we compute  $p_{\text{tx}}^{(itr)}$ as the objective value given in (\ref{eq:p4-object}).
The algorithm  stops when $\Delta p_{\text{tx}}=p_{\text{tx}}^{(itr-1)}-p_{\text{tx}}^{(itr)}\le \Delta p$ holds, with $p_{\text{tx}}^{(0)}=\infty$ by default and $\Delta p>0$ denoting a given stopping threshold. 
The above alternating optimization based  algorithm for  (P4)   is summarized in Table \ref{Tabel:Alternating}, denoted as Algorithm $3$.  
\begin{table}[t!]
\begin{center} 
\caption{Algorithm for (P4).} \scriptsize{
 \hrule\vspace{0.1cm} 
\textbf{Algorithm $3$}
\hrule 
\begin{itemize}
\item[a)]  Initialize $itr=1$, $\Delta p>0$ and $\Delta p_{\text{tx}}=p_{\text{tx}}^{(0)}=\infty$. Initialize  $\tau_{1}=\tau$ and  $x_{n,1}=x_{n}^{*}$, $\forall n \in \cal N$. 
Initialize $\tau_q=0$ and $x_{n,q}=\min\{\hspace{.7mm}\max\{\big({r_n r_{\textnormal{tx}}+w^2 h_n^2}\big)/r_{\textnormal{tx}},\hspace{.7mm} \underline{x}_n \}, \hspace{.7mm} \overline{x}_n\}$, $\forall n \in \cal N$, $\forall q \neq 1$. 
\item[b)]  {\bf While} $\Delta p_{\text{tx}}>\Delta p$ \textbf{do}: 
\begin{itemize}
\item[$\bullet$] Solve (P4) over $\tau_q$'s, $\forall q \in \cal Q$, assuming that the rest of variables are all fixed. Update  $\tau_q$'s as the optimal solution to the resulting problem.
\item[$\bullet$] \textbf{For} $q=1$ to $q=Q$ \textbf{do} : 
\begin{itemize}
\item[--] Solve (P4$-$q) using similar algorithm as Algorithm $1$. Update  $x_{n,q}$'s, $\forall n \in \cal N$, as the  solution to (P4$-$q). 
\end{itemize}
\item[$\bullet$] Compute   (\ref{eq:p4-object}) and  save the obtained value  as    $p_{\text{tx}}^{(itr)}$. Accordingly, set  $\Delta p_{\text{tx}}=p_{\text{tx}}^{(itr-1)}-p_{\text{tx}}^{(itr)}$.
\item[$\bullet$] Set $itr=itr+1$.
\end{itemize}
\item[d)] Return  $x_{n,q}$'s  and $\tau_q$'s as the solution to (P4).
\end{itemize}
\hrule  \label{Tabel:Alternating} }
\end{center} \vspace{-2mm}
\end{table}  
Note that the convergence of Algorithm 3 for (P4) is ensured since the objective value of (P4), i.e., $p_{\text{tx}}^{(itr)}$, is non-increasing  over iterations, while the constraints in (\ref{eq:Const1-p4}) and (\ref{eq:Const2-p4}) are all satisfied at each iteration.  
\subsection{Numerical Example}
We consider the  same system setup as that in Section  \ref{subsec:NumericalExample-Performnce}. We set $\underline{x}_n=1\Omega$ and $\overline{x}_n=100\Omega$, $\forall n\in \cal N$.   For Algorithm $3$ in the case with time sharing,  we  set $\Delta p=10^{-3}$.  

Figs. \ref{fig:Region-comparision}(a), \ref{fig:Region-comparision}(b), and \ref{fig:Region-comparision}(c)  show  power regions  ${\cal R}_{\text{without-TS}}$ versus ${\cal R}_{\text{with-TS}}$ given by (\ref{eq:Reg_WiOu}) and (\ref{eq:Reg_With}), respectively, under three different resonant angular frequencies of  $w=14.2\times 10^6$rad/sec, $w=42.6 \times 10^6$rad/sec, and $w=127.8\times 10^6$rad/sec, respectively. For the purpose of exposition, we consider only the two-user case for Fig. \ref{fig:Region-comparision} by assuming  receiver $3$ is disconnected from its load, i.e., $s_3=0$ (see Fig. 2). 
It is observed that  ${\cal R}_{\text{with-TS}}$ is always larger than ${\cal R}_{\text{without-TS}}$, as expected. 
It is also observed that the power region difference becomes less significant as the operating frequency decreases. 
This  result is  explained as follows. When $w$ is sufficiently small, the power delivered to each load $n$ in the case without time sharing, given in (\ref{eq:PL}), can be approximated as $p_{n} \approx |v_{\text{tx}}|^2 w^2 h_n^2 x_n(r_n+x_n)^{-2}/ (2 r_{\text{tx}}^2)$, from which it follows that there is no evident  coupling effect among receivers.  
In this regime, given load resistance $x_n$, the power received by load $n$ does not depend on  whether the other receivers are connected to their loads or not.  Thus, time sharing is less effective and hence cannot enlarge the power region over that  without time sharing. Last, note that since Algorithm $3$ for (P4) in general obtains a suboptimal solution,  ${\cal R}_{\text{with-TS}}$ shown in Fig. \ref{fig:Region-comparision} is only an   achievable power region under the time-sharing scenario.

Next, we consider again the case with all three users in Fig. 2. We also fix $w=42.6\times 10^6$rad/sec and $\underline{p}_1=\underline{p}_2=5$W.  
Fig. \ref{fig:Simul-alg3} compares the  transmitter power $p_{\text{tx}}$ obtained  using  Algorithms $1$ and $3$ over $\underline{p}_3$, with $0<\underline{p}_3 \le 55.9$W, where Algorithm $3$ takes at most $4$ iterations to converge. 
It is  observed that Algorithm $3$  achieves lower  $p_{\text{tx}}$ than Algorithm $1$ over all values of $\underline{p}_3$.
This result is expected due to the fact that time sharing  provides  extra degrees of freedom for multiuser power transmission scheduling. Consequently, the time allocation and  load resistance for the receivers can be jointly optimized   to further reduce the transmitter power as compared to the case without time sharing when only load resistance is optimized.  

\begin{figure}[t]
	\centering
	\includegraphics[width=12cm]{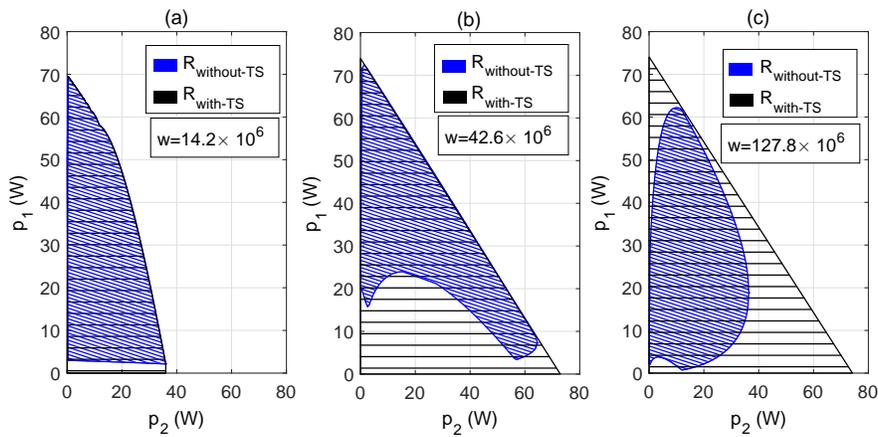}\vspace{-2mm}
	\caption{Power regions without versus with time sharing.} 
	\label{fig:Region-comparision}
\end{figure}

\begin{figure} [t]
	\centering
	\includegraphics[width=12cm]{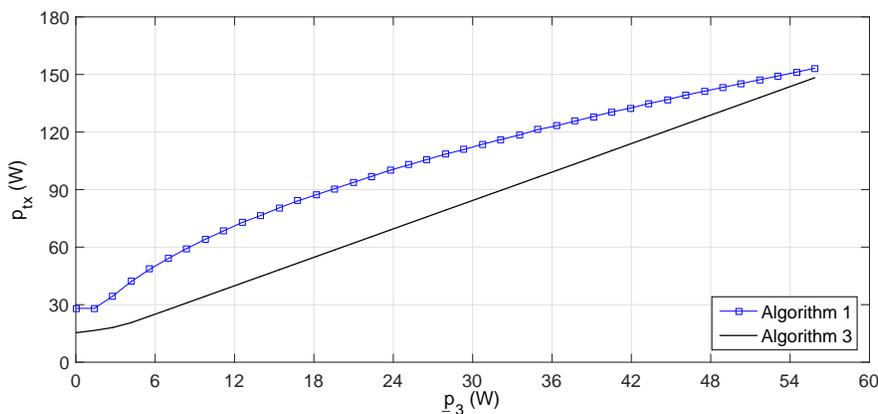} \vspace{-7mm}
	\caption{Performance comparison between Algorithm $1$ (without time sharing) versus Algorithm $3$ (with time sharing).}  
	\label{fig:Simul-alg3} 
\end{figure} 
\section{Conclusion}
In this paper, we have studied  a point-to-multipoint   WPT system via MRC. 
We derive closed-form expressions for the input
and output power in terms of the system parameters for arbitrary number of receivers. 
Similar to
other multiuser wireless applications such as those in wireless communication
and far-field microwave based WPT, a near-far fairness issue is revealed in our considered MRC-WPT system.
To tackle this problem, we propose a centralized charging control  algorithm for jointly optimizing the receivers' load resistance to minimize the transmitter power subject to the given load power  constraints. For ease of practical implementation, we also propose a distributed algorithm for receivers to iteratively adjust
their load resistance based on local information and one-bit feedback
from each of the other receivers. We show by simulation that the distributed algorithm performs very  close to the centralized
algorithm with a finite number of iterations.
Last, we characterize the achievable multiuser  power regions for the loads without and with time sharing and compare them through numerical examples. 
It is shown that time sharing  can help further mitigate the near-far  issue by  enlarging the achievable power region  as compared to the case without time sharing.
As a concluding remark, we would like to point out that MRC-WPT  is a promising new research area in which many tools from signal processing and optimization can be  applied to devise innovative solutions, and we hope that this paper will open up an  avenue for more future works along this direction.  
\appendix
\subsection{Proof of Proposition \ref{Prop:3} }
From  (\ref{eq:PT}), it follows that
\begin{align} \label{eq:Proof-P3-2}
\dfrac{\partial p_{\text{tx}}}{ \partial x_n}=\dfrac{|v_{\text{tx}}|^2}{2} \dfrac{w^2 h_n^2 (r_n+x_n)^{-2}}{\left(r_{\text{tx}}+w^2 \sum_{k=1}^{N} h_k^2\left(r_k+x_k\right)^{-1} \right)^2},
\end{align}
where it can be easily verified that  $\partial p_{\text{tx}}/ \partial x_n>0$ over $x_n>0$. This means that $p_{\text{tx}}$  strictly increases over $x_n>0$. The proof of Proposition \ref{Prop:3} is thus completed.
\subsection{Proof of Proposition \ref{Prop:4} }
From  (\ref{eq:PL}), it follows that for $m\neq n$, 
\begin{align} \label{eq:Proof-P4-2}
\hspace{-2mm}\dfrac{\partial p_{m}}{\partial x_n}=\dfrac{|v_{\text{tx}}|^2}{2}\dfrac{2w^4 x_mh_m^2 h_n^{2}\left(r_m+x_m\right)^{-2} \left(r_n+x_n\right)^{-2}}{\left(r_{\text{tx}}+w^2 \sum_{k=1}^{N} h_k^2(r_k+x_k)^{-1} \right)^3}, \hspace{-1mm}
\end{align}
where it can be easily verified that  $\partial p_{m}/ \partial x_n>0$ over $x_n>0$. This means that $p_m$, $m\neq n$,  strictly increases over $x_n>0$.
 Similarly, for $m=n$, from (\ref{eq:PL}), it follows that 
\begin{align} \label{eq:Proof-P4-3}
&\dfrac{\partial p_{n}}{\partial x_n}=\dfrac{|v_{\text{tx}}|^2}{2}\dfrac{w^2 h_n^2\left(r_n+x_n\right)^{-3}}{\left(r_{\text{tx}}+w^2 \sum_{k=1}^{N} h_k^2(r_k+x_k)^{-1} \right)^3} \bigg(w^2h_n^2 \nonumber \\
&\hspace{38mm}+\left(r_{\text{tx}}+\phi_n\right) \left(r_n-x_n\right) \bigg),
\end{align}
where  $\phi_n=w^2 \sum_{k \in {\cal N}\setminus \{n\}} h_k^2(r_k+x_k)^{-1}$. It can be easily verified that  $\partial p_{n}/ \partial x_n>0$ over $0<x_n<\dot{x}_n$,   with $\dot{x}_n>0$ given in  (\ref{eq:X_Star}), and $\partial p_{n}/ \partial x_n<0$ over $x_n>\dot{x}_n$. The proof of Proposition \ref{Prop:4} is thus completed.  
\subsection{Proof of Proposition \ref{Prop:5} }
Since  $\partial p_{\text{sum}}/\partial x_n=\sum_{k=1}^{N}\partial p_{k}/\partial x_n$, we can easily prove  Proposition \ref{Prop:5}  using the same argument  for  the proof of  Proposition \ref{Prop:4}. The detail is thus omitted for brevity.
\subsection{Proof of Proposition \ref{Prop:6} }
From  (\ref{eq:eta}), it follows that
\begin{align} \label{eq:Proof-rh}
&\hspace{-2.2mm}\dfrac{\partial \rho}{\partial x_n}=  \dfrac{|v_{\text{tx}}|^2}{2}\dfrac{w^2 h_n^2\left(r_n+x_n\right)^{-4} }{\left(r_{\text{tx}}+w^2 \sum_{k=1}^{N} h_k^2(r_k+x_k)^{-1} \right)^2} \bigg(2r_n \varphi_nx_n \nonumber \\ 
&\hspace{-2.2mm}+r_nw^2h_n^2+x_n^2\left(\varphi_n\hspace{-0.1mm}-\hspace{-0.1mm}\phi_n\hspace{-0.1mm}-\hspace{-0.1mm}r_{\text{tx}}\right)\hspace{-0.1mm}+\hspace{-0.1mm}r_n^2\left(\varphi_n\hspace{-0.1mm}+\hspace{-0.1mm}\phi_n\hspace{-0.1mm}+\hspace{-0.1mm}r_{\text{tx}} \right) \hspace{-1mm} \bigg).\hspace{-1.2mm}
\end{align}
Accordingly, it can be verified that when  $\varphi_n-\phi_n-r_{\text{tx}} \ge 0$, $\rho$  strictly increases over $x_n>0$. Otherwise, if $\varphi_n-\phi_n-r_{\text{tx}}< 0$, then $\rho$ increases over $0<x_n<\dddot{x}_n$ due to the fact that $x_n^2(\varphi_n-\phi_n-r_{\text{tx}})+2r_n \varphi_nx_n+r_n w^2 h_n^2+r_n^2(\varphi_n+\phi_n+r_{\text{tx}})<0$ and the rest of terms in the right-hand side of (\ref{eq:Proof-rh}) are all  positive; similarly $\rho$ declines over $x_n>\dddot{x}$. The proof of Proposition \ref{Prop:6} is thus completed.
\subsection{Impedance Characterization  of EM Coils} \label{App:Paramters}
As shown in Fig. \ref{fig:2Coil-Generic},  we consider two circular EM coils, indexed by $i$, $i\in \{1,2\}$, in the free space (no external electric and/or magnetic fields exist). 
Without loss of generality, we assume that the center of EM coil $1$ is located at the origin, i.e., $(x=0,y=0,z=0)$, and its  surface  normal vector  is given by  $\vec{n}_1=\vec{z}$. 
On the other hand, we assume that the center of EM coil $2$ is located at $(x=x^\prime, y=y^\prime, z=z^\prime)$ and its surface  normal vector is given by  $\vec{n}_2=n_{x,2}\vec{x}+n_{y,2}\vec{y}+n_{z,2} \vec{z}$,  with  $\sqrt{n_{x,2}^2+n_{y,2}^2+n_{z,2}^2}=1$. 
As shown in Fig. \ref{fig:Coil-Example}, we assume that each EM coil $i$ consists of $b_{i}$ closely wound turns of round shaped wire, where the inner radius of coil is denoted by $e_{\text{inner},i}>0$, while  the outer radius is denoted by  $e_{\text{outer},i}>e_{\text{inner},i}$. 
\begin{figure}[t!]
\centering
\includegraphics[width=11cm]{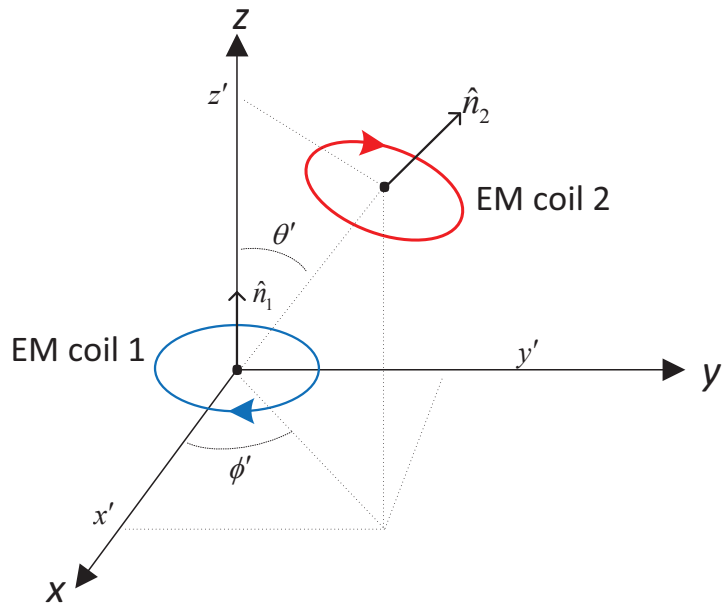}
\caption{A system of two circular EM coils.} 
\label{fig:2Coil-Generic} 
\end{figure}

\begin{figure}[t!]
\centering
\includegraphics[width=11cm]{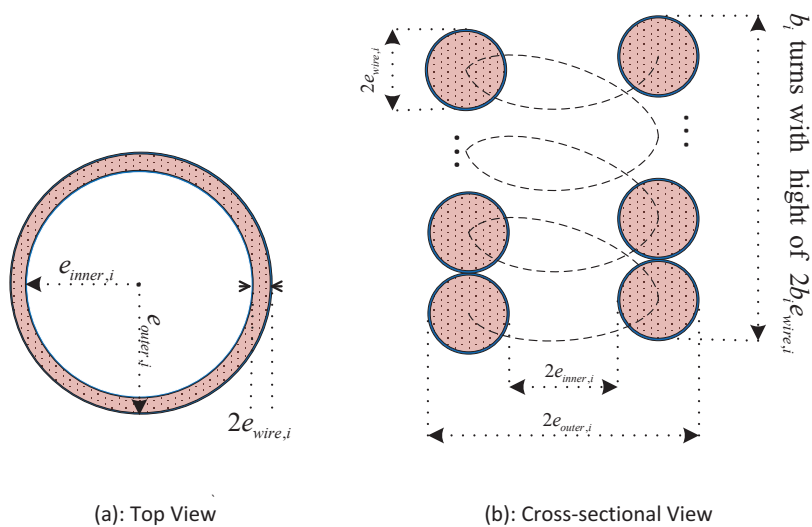} 
\caption{Circular EM coil.}  
\label{fig:Coil-Example}
\end{figure}
Accordingly, the average radius of each EM coil $i$ and the radius of the wire used to build  this coil  are obtained as $e_{\text{ave},i}=(e_{\text{outer},i}+e_{\text{inner},i})/2$ and $e_{\text{wire},i}=(e_{\text{outer},i}-e_{\text{inner},i})/2$, respectively.
Let $r_{i}$ and $l_{i}$ denote  the resistance and self-inductance of  each EM coil $i$.
Given $e_{\text{wire},i}\ll e_{\text{ave},i}$, i.e., the wire is much thinner than the average radius, which is practically valid, we thus have \cite{Chen}:
\begin{align}
r_{i}&=\dfrac{2 \sigma_i b_i e_{\text{ave},i}}{e_{\text{wire},i}^2}, \\
l_{i}&=b_i^2 e_{\text{ave},i} \mu \big(\ln(\dfrac{8 e_{\text{ave},i}}{e_{\text{wire},i}}) -2\big),
\end{align}
where $\sigma_i$ is the resistivity of the wire used in EM coil $i$ and  $\mu=4 \pi \times 10^{-7}$N/A$^2$, which denotes the magnetic permeability of the air.  Let $h$ denote the mutual inductance between the two EM coils. 
By assuming  $d \triangleq \sqrt{{x^\prime}^2+{y^\prime}^2+{z^\prime}^2}\gg e_{\text{ave},1},e_{\text{ave},2} $, i.e., the distance between  the two EM coils is much larger than their average radiuses,  we  have \cite{Cheng}:
\begin{align}
\hspace{-1mm}h= \hspace{-1mm}- \dfrac{\pi \mu b_1 b_2  e_{\text{ave},1}^2 e_{\text{ave},2}^2}{4 d^3} &\bigg(3 \cos(\theta^\prime) \sin(\theta^\prime)\cos(\phi^\prime) n_{x,2} \nonumber \\
&+ 3 \cos(\theta^\prime)\sin(\theta^\prime)\sin(\phi^\prime) n_{y,2}\nonumber \\
&+(2\cos^2(\theta^\prime)-\sin^2(\theta^\prime)) n_{z,2} \bigg),  \hspace{-1mm}
\end{align}
where $\theta^{\prime}=\cos^{-1}(z^{\prime}/d)$ and $\phi^{\prime}=\tan^{-1}(y^{\prime}/x^{\prime})$. 


\end{document}